\documentclass[11pt]{article}

\usepackage{amsmath}
\usepackage{amsthm}
\usepackage{algorithm}
\usepackage{algorithmic}
\usepackage{fullpage}

\newcommand {\E}{\mathbf E}

\renewcommand {\P}{\mathcal P}

\newcommand{\F}{\mathcal{F}}

\newcommand{\R}{{\mathcal R}}
\newcommand{\M}{\mathcal M}
\newcommand{\OPT}{{\tt {OPT}}}
\newcommand{\WMP}{{\mathcal{A}}}
\newcommand{\DD}{{\mathcal{D}}}
\newcommand{\N}{{\mathcal{N}}}
\newcommand{\junk}[1]{}
\newcommand{\U}{{\bf U}}
\renewcommand{\u}{{\mathbf u}}
\newcommand{\Z}{{\mathcal Z}}
\newcommand{\G}{\mathcal G}
\renewcommand{\vec}[1]{{\mathbf #1}}
\newcommand{\C}{{\mathcal C}}
\newcommand{\expert}{{\tt EQ}}

\newtheorem{theorem}{Theorem}[section]

\newtheorem{lemma}[theorem]{Lemma}

\newtheorem{observation}{Observation}

\newtheorem{definition}[theorem]{Definition}

\newcommand{\fullversion}[1]{{#1}}
\newcommand{\shortversion}[1]{{}}

\begin{document}

\title{Optimal Auctions via the Multiplicative Weight Method}
\author{
Anand Bhalgat\footnote{University of Pennsylvania. Supported in part by NSF grant IIS-0904314. Email: {\em bhalgat@seas.upenn.edu.}} \and Sreenivas Gollapudi\footnote{Microsoft Research Search Labs. Email: {\em sreenivas.gollapudi@microsoft.com}.} \and Kamesh Munagala\footnote{Duke University. Supported by an Alfred P. Sloan Research  Fellowship, a gift from Cisco, and by NSF via grants CCF-0745761, CCF-1008065, and IIS-0964560. Part of this work was done while the author was visiting Twitter, Inc. Email: {\em kamesh@cs.duke.edu.}}}

\maketitle

\begin{abstract}
We show that the multiplicative weight update method provides a simple recipe for designing and analyzing optimal Bayesian Incentive Compatible (BIC) auctions, and reduces the time complexity of the problem to pseudo-polynomial in parameters that depend on single agent instead of depending on the size of the joint type space. We use this framework to design computationally efficient optimal auctions that satisfy ex-post Individual Rationality in the presence of constraints such as (hard, private) budgets and envy-freeness. We also design optimal auctions when buyers and a seller's utility functions are non-linear.  Scenarios with such functions include (a) auctions with ``quitting rights'', (b) cost to borrow money beyond budget, (c) a seller's and buyers' risk aversion.
Finally, we show how our framework also yields optimal auctions for variety of auction settings considered in~\cite{CDW_EC12,AlaeiFHHM12,CDW_STOC12,CDW_FOCS12,CDW_SODA13}, albeit with pseudo-polynomial running times.
\end{abstract}


\section{Introduction}
In this paper, we present a simple framework for designing Bayesian Incentive Compatible (BIC) auctions. This model of auction design was pioneered by Myerson~\cite{myerson}. There is a common prior over the private valuations of all buyers, and the mechanism has to enforce incentive compatibility in expectation over the revealed types of other buyers (where these types are assumed to be drawn from the prior), and the randomness introduced by the mechanism. The optimal BIC auction design problem is purely computational - in most cases, it is possible to design an optimal mechanism in time polynomial in the support of the {\em joint type space} of the priors of all buyers. As the support can be exponential in the number of buyers, the question arises: {\em Can we design optimal auctions whose computational complexity is polynomial in the number of buyers, items, and parameters that only depend on individual buyers?} This was first posed as a computational problem in~\cite{vince}.

Though this question was answered in the affirmative in economic literature~\cite{myerson,border} for simple settings such as multi-unit auctions, it remained largely open in more general settings. An impressive  sequence of recent papers~\cite{HartlineL,CDW_EC12,AlaeiFHHM12,CDW_STOC12,CDW_FOCS12,CDW_SODA13} answer this question in the affirmative under a wide variety of allocation constraints, with both single and multi-dimensional types for each buyer. These algorithms are based on reducing the revenue maximization problem to a convex combination over analogous welfare maximization problems on the same allocation constraints via convex programming techniques, hence reducing an exponential time computation over the joint type space to polynomial time in the type space of a single buyer.

\paragraph{\bf Our Contributions}
In this paper, we show that the classical multiplicative weight update method~\cite{Y95,PST95,AHK05} yields a simple framework for designing optimal auctions along with an elementary analysis.  This framework, presented in Section~\ref{sec:framework}, is based on writing an {\em explicit program} connecting ex-interim variables to ex-post variables (allocations, payments, utilities, etc). It then reduces the revenue maximization problem to single-bidder optimization problems using the dual (Lagrangian) variables in the explicit program, hence reducing the running time to depend only on single-buyer input parameters. Its key advantages stem from the nature of the problem formulation that {\em explicitly encodes the optimal mechanism's action in each realization}:
\begin{itemize}
\item It is extremely simple to state and generalize.
\item In contrast with previous techniques, we do not have to argue about feasibility of ex-interim solutions or explicitly convert them to ex-post allocations. The technique is ``one-shot" yielding the mechanism directly from the dual variables.
\item The ellipsoid based techniques developed in the previous work are sensitive to the geometry of the feasible ex-interim allocation region defined by the separation oracle; when this region is not convex, such as when the welfare maximization problem can only be solved approximately ~\cite{CDW_SODA13}, the optimization becomes complex. In contrast, the multiplicative weight update based techniques do not face any such hurdle (as long as the oracle subproblem can be solved) and yield a relatively simple approach to design mechanisms in such conditions (see Appendix~\ref{sec:general_allocation_constraint}).

\end{itemize}
Our approach can be used to derive recent results (mentioned above) on optimal auction design with multi-dimensional type spaces and allocation constraints~\cite{CDW_EC12,AlaeiFHHM12,CDW_STOC12,CDW_FOCS12,CDW_SODA13} (see Appendix~\ref{sec:general_allocation_constraint}), albeit with running times that are pseudo-polynomial in single-bidder parameters (instead of fully polynomial). This is however {\em not} the main focus of the paper.
We focus on  designing optimal mechanisms for several new settings where the only known previous results were approximation algorithms. Though previous techniques can be adapted to solve these variants, we illustrate how our technique provides a direct approach to such auction design problems. The settings we consider involve designing mechanisms for the following scenarios:
\begin{itemize}
\item Auctions with {\em hard payment constraints} such as budgets;
\item Making the  BIC mechanism and utility model more robust to the {\em common prior assumption}, by enforcing individual rationality ex-post (in presence of hard payment constraints), and
\item {\em Non-linear utility functions} for seller and buyers, which includes scenarios such as risk averse seller and/or buyers, buyers with quitting rights, and cost to borrow money beyond budget.
\end{itemize}
We discuss these aspects briefly next, before presenting our results.


\smallskip
\noindent
{\em Ex-post Rationality and Payment Constraints:} A criticism of Bayesian Incentive Compatibility (BIC) is that it is not robust to the common prior assumption on all buyers; in particular, if a buyer is agnostic to this prior information, the allocations and payments may make no sense to him. See for instance~\cite{morris,mook} for a discussion of this aspect, as well as the discussion implicit in~\cite{CDW_EC12}. One simple constraint to make the mechanism more robust is to ensure that the voluntary participation (or individual rationality (IR)) constraint always holds: Given a buyer's revealed type, its allocation and payment in {\em any} outcome of the mechanism yield it non-negative utility for that revealed type. This is a very intuitive constraint, since a buyer will participate even when he has no knowledge of the prior. In several settings, a BIC mechanism can indeed be converted into one where this constraint holds as well. A notable exception to this happens when there are {\em payment constraints} on the buyers.

We consider BIC mechanism design in settings where there are general ex-post constraints not only on allocations, but also on the {\em payments} buyers make. The most important example is constraints on payments made by individuals, {\em i.e.}, budget constraints; another example is fairness constraints on allocations and payments across buyers, such as the popular envy-freeness constraint. The presence of these ex-post payment constraints reduces the participation risk from an individual buyer's perspective and encourages their participation.

Consider for instance the case where buyers have (possibly private) {\em hard budget constraints}. There is a fairly long line of work~\cite{che1,che2,laffont,border,vohra} in economics literature on optimal BIC auctions with either public or private budget constraints. It was realized early on that budgets make the problem significantly harder even when they are publicly known. One crucial difference is that in the presence of payment constraints, a BIC mechanism can be implemented as an all-pay mechanism where buyers submit payment with bids (regardless of final allocations); this is ex-interim IR, but cannot be converted to a mechanism that is ex-post IR.  As far as we are aware, the design of optimal auctions with ex-post IR constraints has not been considered in literature before. 

\smallskip
\noindent
{\em Non-linear Utilities:} In many auction design problems, the buyer's utility is not linear in his allotment or payment. Consider, for example, auctions with {\em quitting rights} for buyers~\cite{cj}. Unlike the traditional BIC model where the allocations and payments are binding ex-post (so that the utility is linear in each scenario and only quasi-linear in expectation), in the quitting rights model, a buyer can decide to not participate (or quit) ex-post, so that the utility is quasi-linear in each outcome of the mechanism, and hence non-linear. The presence of quitting rights is important for a mechanism, as in variety of settings, the outcome cannot be legally enforced on buyers and it remains a buyer's choice to accept the outcome {\em ex-post}. 

We also consider a setting where a buyer incurs an arbitrary monotone cost to arrange for payments beyond his budget~\cite{DHW11}; this includes scenarios where a buyer can take a loan for additional payment, but he incurs an extra cost for the loan set up and the interest. Lastly, we consider risk-aversion of buyers and the seller~\cite{EF99}. Again, previous techniques do not address such constraints.

\paragraph{\bf Informal Problem Statements and Our Results}
We present {\em efficient algorithms for optimal BIC auctions} that satisfy ex-post IR in the presence of constraints on both payment and allocation in fairly general settings. We also design the {\em optimal auctions} in many models where buyers have {\em non-linear utility over allocation and payment}.  

We assume the buyer's type is discrete, and there are $n$ buyers and $m$ items. The running time we obtain  are polynomial in the $n, m, \frac{1}{\epsilon}$, and the size of the type space of each buyer. It also depends polynomially on the ratio of largest to smallest valuations, and the ratio largest to smallest marginal prior probabilities for a single buyer.
\begin{itemize}
\item  {\em Multi-unit auction with arbitrary valuations} (Section~\ref{sec:multi-unit-budget}): There are $m$ homogeneous items. Each buyer is associated with an arbitrary valuation function over the number of items, parameterized by his type. Each buyer can also have a hard (possibly private) budget constraint, and we seek ex-post IR mechanisms.

\item {\em Non-linear utility functions} (Sec.~\ref{sec:non-linear}): We consider the multi-unit auction setting when the private utility of a buyer is an arbitrary function of allocation and payment, which is parametrized by a discrete type. As special cases, we consider the quitting rights model~\cite{cj} described above, as well as the setting where an agent can borrow money at a cost~\cite{DHW11}.  Finally, we allow the seller to have utilities on payments received~\cite{EF99,BCK12}. Our setting allows positive transfers, and do not need individual rationality as a constraint for optimality (both of which can raise more utility for a risk-averse seller).

\item {\em Multi-item auctions} (Appendix~\ref{sec:multi-item}): There are $m$ divisible items (for instance, ad categories in an adword auction). Each buyer's type is multi-dimensional but additive across items. Given any vector of revealed types, the allocations and payments must satisfy an arbitrary set of polyhedral constraints. As before, each buyer can have (potentially private) hard budget constraint, and we require IR to hold ex-post.

In the absence of other constraints on allocations and payments, we enforce ex-post envy-freeness constraints where a buyer should not derive more utility from the allocation and payment made to another buyer.

\item {\em Procurement auctions} (Appendix~\ref{sec:budget-feasible}): The agents have private costs on procuring items. The utility of the auctioneer is linear in the subset of items he procures. The auctioneer has a budget constraint on payments. This is a special case of {\em budget feasible mechanisms}~\cite{Singer,Chen} with linear utility functions.
\end{itemize}

We note that all results in this work easily extend when the distributions are correlated, as long as there is black box access to samples from the joint type space. \fullversion{(See Appendix~\ref{sec:correlated} for details.)}\shortversion{ Details are deferred to the full version.} Finally, though it is not the main point of the paper, we also note that this approach also derives all recent results on optimal auction design with multi-dimensional type spaces and allocation constraints~\cite{CDW_EC12,AlaeiFHHM12,CDW_STOC12,CDW_FOCS12,CDW_SODA13} (see Section~\ref{sec:framework} and Appendix~\ref{sec:general_allocation_constraint}), albeit with pseudo-polynomial running times.

\paragraph{\bf Overview of the Technique} The general framework to derive auctions using the multiplicative weight method is summarized in Section~\ref{sec:framework}. A key ingredient of this framework is to explicitly encode the {\em mechanisms's action} (allocation and payment for each buyer) {\em in each realization}, and connect them to the set of {\em holistic variables} using  the {\em linearity of expectation} constraints.
The former are exponentially many in number, but are connected to the latter via polynomially many linearity of expectation constraints.
Our main insight is to use the multiplicative weight update (or Lagrangean) method
to decide feasibility of these linearity of expectation constraints, and show that the oracle sub-problem that this method requires has a very simple structure - it decouples into a linear program over holistic variables, and into a per sample (or type vector) optimization problem of a certain form over the mechanism's {\em action variables}.
Even though the latter is over exponentially many type vectors, it can be estimated by solving the per-sample optimization subproblem over polynomially many samples. The per sample optimization subproblem is problem-specific, yet it is an easy exercise to decide when this problem will admit to polynomial time solutions. Assuming it does, the framework yields a pseudo-polynomial time algorithm to computer a $\epsilon$-BIC mechanism that is an additive $\epsilon$ approximation to the optimal revenue.

We note, the classical multiplicative weight update framework~\cite{AHK05} looks at the per-constraint violation by the {\em average-of-all-rounds} solution; we require (and establish) a slight generalization to this framework (see Section~\ref{sec:ahk}):
\begin{itemize}
\item The framework can handle different linear constraints in each round.
\item We consider the average violation of the $i$th constraint (for each $i$) by the solution to the oracle sub-problem for a uniformly chosen round. With this change, the framework no longer requires the underlying solution space to be convex.
\end{itemize}
In comparison with previous work, our optimization technique is based on the Lagrangean of an {\em explicitly written program} instead of trying to directly characterize the space of primal interim allocations rules.
Our running times are pseudo-polynomial in the ratio of largest to smallest marginal prior probabilities. We leave improving this dependence without sacrificing any BIC or optimality properties as future work.

\paragraph{\bf Related Work}
Though many auctions implemented in practice incorporate budget constraints (for instance, ad allocation auctions), most theoretical results in the absence of a Bayesian assumption show impossibility of designing auctions with specific properties. For instance, recent work~\cite{DLS09,GMP12} shows the impossibility of designing auctions that are Pareto-optimal in terms of welfare even in multi-unit settings. Furthermore, Nisan~\cite{nisan} show that competitive equilibria need not exist even when agents have additive valuations for items. The lone exceptions are positive results for simple multi-unit settings~\cite{DLS09,GMP12,BCMX10,abrams} with unlimited supply of the item. In view of these impossibility results, it is more natural to consider the Bayesian setting; however all results so far~\cite{che1,che2,laffont,border,vohra,CDW_EC12,CDW_FOCS12} only satisfy individual rationality in expectation over the priors. We present (what we believe are) the first efficient algorithms for designing optimal auctions that satisfy individual rationality {\em ex-post} in a wide variety of situations.

We finally note that there has been a sequence of recent work on designing approximation algorithms for auctions in the Bayesian setting. For instance, the results in~\cite{CHMS10,BGGM10,BCMX10,ChawlaMM11} and subsequent papers show constant factor approximations to optimal auctions under a stronger model called dominant strategy incentive compatibility; the results in~\cite{Singer,Chen} show constant factor approximations to the welfare of procurement auctions subject to a budget constraint (what they term {\em budget feasible mechanisms}); and the results in~\cite{BCK12} considers the seller's as well as buyers' risk-aversion. In this paper, our focus is on designing optimal (or near-optimal) auctions and not on approximation algorithms.

\section{Preliminaries}
\label{sec:prelim}
We describe notations used in the paper: There are $n$ buyers and a risk-neutral seller. The type space for buyer $i$ is $T_i$, and we denote the individual type of buyer $i$ by $t_i$.
We use $\vec{t}$ to indicate the vector of reported types of all buyers in a realization of the mechanism. We also abuse the notation slightly to use $t_i$ to indicate buyer $i$'s type in $\vec{t}$. The seller maintains a prior distribution over the types of the buyers. The probability distribution of buyer $i$ over its type is given by $\DD_i$ and we use $\DD$ to indicate the joint distribution over types of all buyers.  We use $\mu(\vec{t})$ to denote the density of type vector $\vec{t}$, and $f_i(t_i)$ to denote the probability of type $t_i$ for buyer $i$. For exposition of our results, we assume independence across buyers; however our results easily extend when buyers' distributions are correlated \fullversion{(Appendix~\ref{sec:correlated})}\shortversion{(details deferred to the full version)}. Further, for each auctions setting studied in this work, we use $\OPT$ to denote the expected revenue (or the utility) of the optimal mechanism under the same set of constraints.

\subsection{Bayesian Incentive Compatibility}
A mechanism is a mapping from a vector of revealed types (or bids) to a (possibly randomized) set of allocations and payments for each buyer. For most of this paper, we assume the utility of a buyer is {\em linear}: For any allocation made and payment charged to the buyer, his utility is the difference between the welfare from the allocation and the payment.

\paragraph{\bf Bayesian (ex-interim) Incentive Compatibility} In this model of incentive compatibility, we assume the prior is common knowledge, and buyers are interested in maximizing their expected utility, where the expectation is over the types of the other buyers drawn from their respective priors. More formally, let $U_i(t'_i,t_i)$ denote the expected utility for buyer $i$ when he reveals type $t'_i$ and his true type is $t_i$, where this expectation is over the types of the other buyers  (drawn from their prior distributions) and the randomness in the mechanism. A Bayesian Incentive Compatible (BIC) mechanism has to enforce that:
\begin{description}
\item[Incentive Compatibility (IC):] $U_i(t_i,t_i) \ge U_i(t'_i,t_i)$ for every pair of types $t_i,t'_i\in T_i$, and
\item[Individual Rationality (IR):]  $U_i(t_i,t_i) \ge 0$ for all types $t_i\in T_i$.
\end{description}
A mechanism is said to be $\epsilon$-BIC if it enforces $U_i(t_i,t_i) \ge U_i(t'_i,t_i) - \epsilon$ for every pair of types $t_i,t'_i\in T_i$.

\paragraph{\bf Robust (ex-post) Individual Rationality} The BIC constraint only requires the allocations and payments to satisfy IR ex-interim. In this paper, we consider a sub-class of mechanisms that satisfy a more robust (ex-post) notion of IR: If buyer $i$ reveals type $t_i$, then its allocation and payment should yield non-negative utility for the reported type {\em in each outcome} of the mechanism.
We term this constraint {\em ex-post IR}.

\paragraph{\bf Budgets Constraints} Each buyer is associated with a budget constraint. The budgets can be either publicly known, or it is buyer's private information specified using his type. We denote buyer $i$'s budget by $B_i$, when budgets are public. For private budgets, we use $B_i(t)$ to denote buyer $i$'s budget when his type is $t$. Unless mentioned otherwise, we assume that a buyer's utility of paying more than his budget is $-\infty$.


\subsection{Problem Statements}
 In this work, we primarily focus on designing BIC mechanisms that maximize expected revenue and satisfy IR and budget constraints {\em ex-post.}

\begin{itemize}
\item {\em Multi-unit Auctions (Section~\ref{sec:multi-unit-budget}):} The seller has $m$ identical items, each buyer $i$ is associated with a publicly known valuation function $v_i:\N\times T_i\rightarrow \R^{+}$; his valuation for receiving $k$ items is  $v_{i}(k,t_i)$ when his type is $t_i$. We assume that $v_i(k,t_i)\ge 0$ for each $k\ge 0, i, t_i\in T_i$, and $v_i(0,t_i)=0$ for each $i,t_i\in T_i$. We note, these valuations easily capture scenarios where each buyer wants at most one item.

\item
{\em Non-linear Utility (Section~\ref{sec:non-linear}):} Various multi-unit auction settings that involves buyers as well as the seller with non-linear utility functions are illustrated in Section~\ref{sec:non-linear}.

\item {\em Multi-item Auctions with Divisible Items (Appendix~\ref{sec:multi-item}):} The seller offers $m$ distinct types of divisible items. We assume that buyers have additive valuations: if buyer $i$  of type $t$ receives $x_{ij}$ amount of item $j$, for each $j$, then his total valuation of the allocation is $\sum_{j} x_{ij}v_{ij}(t)$. Further, in addition to budget and ex-post IR constraints, we are given a polyope $\P$ over the set of feasible allocation and payments in a realization of the mechanism.

\item {\em Budget Feasible Mechanisms:} We describe the setting as well as our result in Appendix~\ref{sec:budget-feasible}.
\end{itemize}

\paragraph{\bf Running Times}
We are given a parameter $L$, that for any single buyer, bounds the absolute value of (a) his payment, (b) his valuation for any set of items, and (c) his utility (for any allocation and payment).
Further, we assume that $f_i(t)\ge 1/L$ for each $i,t\in T_i$. Our algorithms produce $\epsilon$-BIC mechanisms that are additive $\epsilon$ approximation to revenue in time polynomial in $n$ (\#buyers), $m$ (\#items),  $\max_i |T_i|$, and the parameters $L$ and $1/\epsilon$.

\newcommand{\D}{\mathcal{D}}
\section{The Multiplicative Weights Update Framework}\label{sec:ahk}
In this section, we present a slight generalization to the multiplicative weights update framework to decide the feasibility of a set of linear constraints over arbitrary spaces. We note its difference compared to the classical framework:
\begin{itemize}
\item The framework can (and needs to) handle different linear constraints in each round.
\item We consider the average violation of the $i$th constraint (for each $i$) by the solution to the oracle sub-problem for a uniformly chosen round (instead of the violation by the average solution of all rounds). With this change, the framework no longer requires the underlying solution space to be convex (Observation~\ref{obs:ahk}).
\end{itemize}
We next illustrate the  framework using the exposition and bounds due to \cite{AHK05}.
We first define the generic problem to decide the feasibility of linear constraints: Given an arbitrary set $P\in \mathbf{R}^s$, and an $r \times s$ matrix $A$,

\begin{center}
\fbox{\begin{minipage}{2.6in}
{\sc LP}$(A,b,P)$:  $\exists x \in P$ such that $A x \ge b$?
\end{minipage}}
\end{center}

The running time of the algorithm is quantified in terms of the {\sc Width} defined as:
$$\rho = \max_i \max_{x \in {P}} |a_i x - b_i|$$
The algorithm assumes an efficient oracle of the form: Let $y\ge 0$ be an $r$ dimensional dual vector.

\begin{center}
\fbox{\begin{minipage}{3.5in}
{{\sc Oracle}$(A, y)$: Solve $C(A, y) = \max\{ y^t A z : z \in P\}$.}
\end{minipage}}
\end{center}

We now describe a generalization of the Arora-Hazan-Kale (AHK) procedure that allows {\em the matrix $A$ and vector $b$ to be different in each step}, albeit with the same number of rows $r$, and same width bound $\rho$. Denote the matrix at time $t$ to be $A_t$, and its $i^{th}$ row by $a_{it}$; denote the vector $b$ at time $t$ to be $b_t$.

\begin{algorithm}
\caption{Generalized AHK Algorithm}
\begin{algorithmic}
\STATE Let $K \leftarrow  \frac{4\rho^2 \log r}{\epsilon^2}$; $y_1 = \vec{1}$
\FOR{$t = 1$ {\bf to} $K$}
\STATE Find $x_t$ using {\sc Oracle}$(A_t, y_t)$.
\IF{$C(A_t, y_t) < y_t^T b_t$}
\STATE Declare {\sc LP}$(A_t,b_t,P)$ infeasible and terminate.
\ENDIF
\STATE {\bf For} $i = 1,2,\ldots,r$, set $M_{it} = a_{it} x_t - b_{it}$ .
\STATE Update $y_{it+1}$ for $i = 1,2,\ldots,r$ as follows: \\
$\qquad y_{it+1} \leftarrow y_{it}(1 - \epsilon)^{M_{it}/\rho}$ if  $M_{it} \ge 0$. \\
$\qquad y_{it+1} \leftarrow y_{it}(1 + \epsilon)^{-M_{it}/\rho}$ if  $M_{it} < 0$. \\
\STATE {\em Optional step:} Normalize $y_{t+1}$ to have unit $l_1$-norm. \\
\ENDFOR
\STATE {\bf Return} $(x_1, x_2, \ldots, x_K)$.
\end{algorithmic}
\end{algorithm}

\junk{
\begin{center}
\fbox{\begin{minipage}{3.25in}
\begin{tabbing}
{\sc Generalized AHK Algorithm} \\
Let $T \leftarrow  \frac{4\rho^2 \log r}{\epsilon^2}$; $y_1 = \vec{1}$ \\
{\bf For} $t = 1$ {\bf to} $K$ {\bf do}: \\
\ \ \ \=  Find $x_t$ using {\sc Oracle}$(A_t, y_t)$. \\
\> If $C(A_t, y_t) < y_t^T b_t$, then declare {\sc LP}$(A_t,b_t,P)$ infeasible and terminate. \\
\> Set $M_{it} = a_{it} x_t - b_{it}$ for $i = 1,2,\ldots,r$. \\
\> Update $y_{it+1}$ for $i = 1,2,\ldots,r$ as follows: \\
\> \ \ \ \= $y_{it+1} \leftarrow y_{it}(1 - \epsilon)^{M_{it}/\rho}$ if  $M_{it} \ge 0$. \\
\> \>  $y_{it+1} \leftarrow y_{it}(1 + \epsilon)^{-M_{it}/\rho}$ if  $M_{it} < 0$. \\
\> \> {\em Optional step:} Normalize $y_{t+1}$ to have unit $l_1$-norm. \\
{\bf Return} $(x_1, x_2, \ldots, x_K)$.
\end{tabbing}
\end{minipage}}
\end{center}
}
The following theorem follows by simply re-working the analysis in~\cite{AHK05}, we omit its proof.
\begin{theorem}
\label{thm:ahk}
If for all $t$,  {\sc LP}$(A_t,b_t,P)$ is feasible, the procedure never declares infeasibility, and the sequence $(x_1, x_2, \ldots, x_K)$ satisfies:
$$ \frac{\sum_t (a_{it} x_t - b_{it})}{K}  + \epsilon \ge 0 \qquad \forall i$$
\end{theorem}
\begin{observation}\label{obs:ahk}
If $P$ is convex and $(A_t, b_t)$ are the same for all $t$, then $x^* = (\sum_t x_t)/K$ is a feasible solution with $A x^* \ge b - \epsilon \vec{1}$. Further, even if $P$ is not convex, the procedure is well-defined and we can interpret $x^*$ as choosing one of the $x_t$ uniformly at random from $\{x_1, x_2, \ldots, x_K\}$.
\end{observation}

\section{Multi-unit Auctions}\label{sec:multi-unit-budget}
In this section, we describe a framework for designing mechanisms by the multiplicative weight update method using multi-unit auctions as an example. We generalize this example to a structural framework in Section~\ref{sec:framework}, and present more instantiations of this framework in later sections. We begin by the case where valuations of different buyers are independent, the only feasibility constraint is on the supply of items, and i.e. at most $m$ items can be allocated, each buyer is associated with a public budget constraint, and the objective is to design an $\epsilon$-BIC mechanism that is (a) feasible, (b) ex-post IR, and (c) satisfies the budget constraints in each realization. The technique to handle private budgets is explained in Section~\ref{sec:private-budget}.

The following theorem states our result for multi-unit auctions:
\begin{theorem}\label{thm:budgets}
For multi-unit auctions with per buyer budget constraints, there exists a polynomial time algorithm to compute a mechanism that is ex-post IR, $\epsilon$-BIC, and is an $\epsilon$-approximation to the revenue of the optimal mechanism (under the same set of constraints).
\end{theorem}
In the rest of the section, we prove Theorem~\ref{thm:budgets}. Consider a realization of the mechanism, let $\vec{t}$ be the reported type vector of all buyers. In the realization, let $x_{iq}(\vec{t})\in\{0, 1\}$ be an indicator variable indicating whether buyer $i$ is allocated $q\in \{0,1,\ldots,m\}$ items when the reported type vector is $\vec{t}$. Let $p_{i}(\vec{t})$ be the payment made by buyer $i$. As vector $\vec{x}(\vec{t}), \vec{p}(\vec{t})$ must be feasible, it should satisfy the constraints $\F(\vec{t})$ shown below.

\[ \begin{array}{rcll}
\sum_{i} qx_{i,q}(\vec{t}) & \le & m & \\
\sum_{q} v_i(q, t_i)x_{iq}(\vec{t}) & \ge & p_{i}(\vec{t}) & \forall i,q\\
x_{iq}(\vec{t}) &\in& \{0,1\} &\forall i,q\\
\sum_q x_{iq} &=& 1 & \forall i \\
p_{i}(\vec{t}) &\in& [0, B_i] &\forall i
\end{array}\]
Let $X_{iq}(t_i)$ be the probability that buyer $i$ is allocated $q$ items when his reported type is $t_i$, and let $P_i(t_i)$ be his expected payment when his type is $t_i$. This probability is over the randomness due to the other buyers' valuations, and due to the randomness in the mechanism. We get
\begin{align}
\label{eq1} f_i(t_i) X_{iq}(t_i) &= \sum_{\vec{t} | t_i \in \vec{t}} \mu(\vec{t}) x_{iq}(\vec{t}) &\forall i, q,t_i\in T_i \\
\label{eq2} f_i(t_i) P_i(t_i) &= \sum_{\vec{t} | t_i \in \vec{t}} \mu(\vec{t}) p_{i}(\vec{t}) &\forall i, t_i\in T_i
\end{align}
The BIC constraint is the following: $\forall i, t_i,t'_i\in T_i$,
\begin{equation}
\label{eq3}\sum_q v_i(q, t_i)X_{iq}(t_i) - P_i(t_i) \ge \sum_q v_i(q, t_i)X_{iq}(t'_i) -P_i(t'_{i})
\end{equation}
Further, we add a constraint to check whether revenue $\OPT$ is feasible. This can be written as the constraint:
\begin{equation}
\label{eq4} \sum_{i,t_i\in T_i} f_i(t_i)  P_i(t_i) \ge  \OPT
\end{equation}

\paragraph{\bf Randomized Mechanisms} An astute reader will object that the set $\F(\vec{t})$ only captures the space of {\em deterministic} mechanisms and {\em an optimal mechanism can take a randomized action for type vector $\vec{t}$}. To capture randomized mechanisms, we have to replace $\F(\vec{t})$ with the {\em convex hull} of the set $\F(\vec{t})$. However, this will only be a technicality - the multiplicative weight update framework described below optimizes a linear function of $\{x(\vec{t}), p(\vec{t})\}$ over this set, and this optimization problem is clearly the same whether we use $\F(\vec{t})$ or its convex hull. Since the same observation can be made of all the problems considered in this paper, we will  use $\F(\vec{t})$ for expositional simplicity.

\subsection{Structural Characterization}
We define the set of feasibility constraints $\C$ as the set of constraints $\F(\vec{t})$ for all type vectors $\vec{t}$, combined with the BIC constraint (\ref{eq3}) and the revenue constraint Eq (\ref{eq4}). Consider deciding feasibility of Equations~\ref{eq1} and~\ref{eq2} subject to the constraints $\C$. This is equivalent to deciding the feasibility of following constraints, referred to as $\expert$:
\[ \begin{array}{rcll}
 X_{iq}(t_i) &=& \sum_{\vec{t} | t_i \in \vec{t}} \mu(\vec{t}) x_{iq}(\vec{t})/ f_i(t_i) &\forall i, q,t_i\in T_i \\
 P_i(t_i) &=& \sum_{\vec{t} | t_i \in \vec{t}} \mu(\vec{t}) p_{i}(\vec{t})/ f_i(t_i) &\forall i, t_i\in T_i
\end{array}\]
Using the multiplicative weight method, the feasibility problem reduces to a sequence of {\sc Oracle}$(\vec{\alpha}, \vec{\beta})$ problems, where $\alpha_{iqt_i}, \beta_{it_i}\ (1\le i\le n, q\in \{0,1,\ldots,m\}, t_i\in T_i)$ are the dual multipliers which could be positive or negative (as the constraints are equality constraints). {\sc Oracle}$(\vec{\alpha}, \vec{\beta})$ asks us to solve, for $\vec{\alpha}, \vec{\beta}$:
\[\begin{array}{rl}
\mbox{Max } & -\sum_{i,t_i\in T_i}  f_i(t_i) \left(\sum_q\alpha_{iqt_i} X_{iq}(t_i)+ \beta_{it_i} P_i(t_i)\right) \\
 &+\sum_{\vec{t}} \mu(\vec{t}) \left(\sum_{i} \left(\sum_q(\alpha_{iqt_i} x_{iq}(\vec{t})) + \beta_{it_i}p_{i}(\vec{t})\right) \right)
\end{array}\]
subject to the following constraints:
\[ \begin{array}{rcll}
\sum_q v_{i}(q, t_i) X_{iq}(t_i) - P_i(t_i) &\ge&  \sum_q v_{i}(q, t_i)X_{iq}(t'_i) -P_{i}(t'_i) &\forall i, t_i,t'_i\in T_i\\
\sum_{i,t_i\in T_i} f_i(t_i) P_i(t_i) &\ge& \OPT \\
\{ x_{iq}(\vec{t}) \} & \in &  \F(\vec{t}) &  \forall \vec{t}
\end{array} \]

\paragraph{\bf Decoupling of the Oracle} Note that {\sc Oracle}$(\vec{\alpha}, \vec{\beta})$ separates into two optimization problems. The first is the poly-size linear program ${\tt LP_{exp}}$:
$$ \mbox{Minimize}  \sum_{i,t_i\in T_i}  f_i(t_i) \left(\sum_q\alpha_{iqt_i} X_{iq}(t_i) + \beta_{it_i} P_i(t_i)\right)$$
subject to the following constraints:
\[ \begin{array}{rcll}
\sum_q v_{i}(q, t_i) X_{iq}(t_i) - P_i(t_i) &\ge&  \sum_q v_{i}(q, t_i)X_{iq}(t'_i) -P_{i}(t'_i) &\forall i, t_i,t'_i\in T_i \\
  \sum_{i,t_i\in T_i} f_i(t_i) P_i(t_i) &\ge& \OPT \\
\end{array}\]
The second is the problem $\WMP(\vec{\alpha}, \vec{\beta})$, which for every type vector $\vec{t}$, solves the following problem:
$$  \mbox{Maximize} \sum_{i}\left(\sum_q \alpha_{iqt_i} x_{iq}(\vec{t})\right) + \beta_{it_i}p_{i}(\vec{t}) $$
subject to the constraints $\F(\vec{t})$.

\paragraph{\bf Solving the Problem $\WMP(\vec{\alpha}, \vec{\beta})$} Now we give a poly-time algorithm to solve the problem $\WMP(\vec{\alpha}, \vec{\beta})$. Observe that, in an optimal assignment, if buyer $i$ is given $q$ items, then  (a) $\beta_{it_i}\le 0\Rightarrow p_{i}(\vec{t}) = 0$, and (b) $\beta_{it_i}> 0\Rightarrow p_{i}(\vec{t})=\min\{B_i, v_i(q, t_i)\}$.

The allocation problem can be solved using dynamic programming as follows: Process buyers in order $n$ to $1$. Let $A[i,k]$ be the value of $\WMP(\vec{\alpha}, \vec{\beta})$ considering buyers $i$ though $n$ and $k$ items. The values of $A[i,k]$ can be computed as follows:
\[\begin{array}{rcl}
A[i,k] &=&\min_{0\le j\le k}(A[i+1, k-j] +j\times\alpha_{ijt_i} F+\max\{\beta_{it_i}\times \min\{B_i, v_i(j,t_i)\}, 0\})
\end{array}\]
As we shall see later, we will require the algorithm to break ties consistently; this can always be achieved by having a state-less implementation of the algorithm.

\paragraph{\bf Using Multiplicative Weight Update Framework} Now we setup the parameters for the multiplicative weight update framework. Recall that $L$ is the largest possible welfare a single buyer can receive from any allocation. Further, $L$ also upper bounds individual buyer's budgets. Then, the width of each equality constraint ($\expert$) is bounded by $\max\{n,L\}=L$.  As $\max_i |T_i| \le L$, the number of constraints in $\expert$ are bounded by $2nmL$. Further, the minimum support for any type for a buyer is bounded by $1/L$. Given any $\epsilon>0$,  choose $\delta=\epsilon/(nmL)$ and $K=O\left(\frac{L^2 \log(nmL)}{\delta^2}\right)$; let $(\vec{\alpha^1},\vec{\beta^1}), (\vec{\alpha^2},\vec{\beta^2}), \ldots, (\vec{\alpha^K}, \vec{\beta^K})$  be the dual variables in rounds $1,\ldots,K$ associated with the multiplicative weight update process. Our mechanism $\M\left(\vec{\alpha}[1\ldots K], \vec{\beta}[1\ldots K]\right)$ is defined below; Lemma~\ref{lem:ahk-multi-unit} establishes its important properties.
\begin{definition}
Let $\vec{t}$ be the reported type vector. The mechanism $\M\left(\vec{\alpha}[1\ldots K], \vec{\beta}[1\ldots K]\right)$ chooses an integer $\ell$ between $1$ and $K$ uniformly at random, and solves the problem $\WMP(\vec{\alpha^{\ell}}, \vec{\beta^{\ell}})$ for type vector $\vec{t}$. It then executes allocations and payments according to the solution of $\WMP(\vec{\alpha^{\ell}}, \vec{\beta^{\ell}})$.
\end{definition}
\begin{lemma}\label{lem:ahk-multi-unit}
(a) The mechanism $\M\left(\vec{\alpha}[1\ldots K], \vec{\beta}[1\ldots K]\right)$ is $\epsilon$-BIC, ex-post IR, (b) a buyer never pays more than his budget, and (c) its expected revenue is at least $\OPT-\epsilon$.
\end{lemma}
\begin{proof}
Using the properties of the multiplicative update framework, Equations~\ref{eq1}, and \ref{eq2} are violated by at most $\epsilon/(nL)$. Since $\max_{i,q,t_i\in T_i}v_{i}(q, t_i) \le L$, the error in each BIC constraint is bounded by $\epsilon/n$. Further, the error in the revenue equation is bounded by $\epsilon$ as $\max_{i,t_i\in T_i} P_i(t_i)\le L$.  This completes the proof of the lemma.
\end{proof}
The only remaining issue is to compute the values of the dual variables efficiently; computing them exactly will require time $\Pi_{i}|T_i|$ which is exponential in $n$.

\subsection{Efficient Computation via Sampling}
\label{sec:rt}
The main hurdle in running multiplicative update framework is that each constraint in $\expert$ is over exponentially many variables. We relax each constraint in $\expert$ by $\delta=\frac{\epsilon}{nmL}$. The relaxed constraints are
\begin{align}
\label{eqn5} X_{iq}(t_i) &\in \frac{\sum_{\vec{t} | t \in \vec{t}} \mu(\vec{t}) x_{iq}(\vec{t})}{f_i(t_i)} \pm \delta  &\forall i,q,t_i\in T_i\\
\label{eqn6} P_i(t_i) &\in \frac{\sum_{\vec{t} | t \in \vec{t}} \mu(\vec{t}) p_{i}(\vec{t})}{f_i(t_i)} \pm \delta &\forall i, t_i\in T_i
\end{align}
If the relaxed constraints are feasible (with respect to feasibility constraints $\C$), then it is easy to check that this implies that the BIC constraints are satisfied to an additive $\epsilon$, and the revenue generated is at least $\OPT - \epsilon$. Further, note that the {\em width} of each of relaxed constraint is bounded by $\rho=mL$. We apply the multiplicative weight update method as follows: In each round $\ell$, given dual weights $(\vec{\alpha}^{\ell},\vec{\beta}^{\ell})$, we do the following. Generate $C=O(nmL^3\log(nmLK)/\epsilon)$ samples from $\DD$; let $S_{\ell}$ be this set.
Let $S_{it_i\ell}\subset S_{\ell}$ be the set of samples in which buyer $i$ has type $t_i\in T_i$. We consider relaxed and sampled version of constraints $\expert$:
\begin{align*}
X_{iq}(t_i) &\in \frac{\sum_{\vec{t}\in S_{it_i\ell}} x_{iq}(\vec{t})}{|S_{it_i\ell}|} \pm \delta &\forall i, q, t_i\in T_i\\
P_i(t_i) &\in \frac{\sum_{\vec{t}\in S_{it_i\ell}} p_{i}(\vec{t})}{|S_{it_i\ell}|} \pm \delta &\forall i, t_i\in T_i
\end{align*}

We now apply the {\sc Oracle} to these relaxed constraints using the dual weight vector $\left(\vec{\alpha}^{\ell}, \vec{\beta}^{\ell}\right)$.
We make a crucial observation:
\begin{observation}
The function {\sc Oracle} reduces to solving $\WMP(\vec{\alpha}^{\ell}, \vec{\beta}^{\ell})$ on these samples.
\end{observation}
We now establish the correctness of our procedure. Let ${\tt LP(\DD)}$ be the feasibility problem (constraints $\C$ and $\expert$) defined on distribution $\DD$ and let ${\tt LP}(S_{\ell})$ for $1\le \ell\le K$ be the feasibility problem defined on $S_{\ell}$, where constraints in $\expert$ are relaxed by $\delta$. We focus on the allocation constraints; the analysis for payment constraints is similar. Let $\{x^{\ell}_{iq}(\vec{t})\}$ denote the solution to the {\sc Oracle} on ${\tt LP}(S_{\ell})$. Let $\{x_{iq}^*(\vec{t})\}$ denote the feasible solution for ${\tt LP(\DD)}$ with no slack.

\begin{lemma}
\label{lem:whp}
With high probability ($\ge 1-1/\mbox{poly}(nmLK/\epsilon)$), for all $1 \le l \le K$, we have:
\begin{align*}
\left|\frac{\sum_{\vec{t}\in S_{it_i\ell}} x^{\ell}_{iq}(\vec{t})}{|S_{it_i\ell}|}-\frac{\sum_{\vec{t}\in \DD|t_i\in \vec{t}}\mu(\vec{t})x^{\ell}_{iq}(\vec{t})}{f_i(t_i)}\right| &\le \delta \\
\left|\frac{\sum_{\vec{t}\in S_{it_i\ell}} x^*_{iq}(\vec{t})}{|S_{it_i\ell}|}-\frac{\sum_{\vec{t}\in \DD|t_i\in \vec{t}}\mu(\vec{t})x^*_{iq}(\vec{t})}{f_i(t_i)}\right| &\le  \delta
\end{align*}
\end{lemma}
\begin{proof}
Since we assumed $f_i(t_i) \ge 1/L$, using Hoeffdings's bounds and union bound, the probability that, for each $i, t_i\in T_i, \ell$, $|S_{it_i\ell}| \in (1\pm \epsilon) C\times f_i(t_i)$ is $1-poly(nmLK/\epsilon)$. Thus $|S_{it_i\ell}|=\Omega(nmL^2K/\epsilon)$ for each $i,t\in T_i,K$.

We now prove the first part. The value of $x^{\ell}_{iq}(\vec{t})$ is determined by the values of dual variables $\alpha^{\ell}, \beta^{\ell}$. From the previous observation, the {\sc Oracle} decomposes into solving  $\WMP(\vec{\alpha}^{\ell}, \vec{\beta}^{\ell})$ on individual samples. The values of $x_{iq}(\vec{t})$ for different values of $\vec{t}$ can be seen as a distribution parameterized by  $\alpha^{\ell}, \beta^{\ell}$, and since the set $S_{\ell}$ is chosen independent of previous rounds and of $\alpha^{\ell}, \beta^{\ell}$, the process of drawing type vectors for $S_{\ell}$ can be seen as drawing independent samples from this distribution. As each $x_{iq}(\vec{t})\in \{0,1\}$ (i.e. a bounded random variable), we can apply Hoeffding's inequality and union bound to prove the first part.

For the second part, note $x^{*}_{iq}(\vec{t})\in\{0,1\}$, and the process of sampling valuation vectors $\vec{t}$ corresponds to sampling from distribution $\{x^{*}_{iq}\}$. The proof now follows using Hoeffding's inequality and union bound.
\end{proof}

As a corollary, suppose ${\tt LP(\DD)}$ is feasible without any slack, then ${\tt LP(S_{\ell})}$ is feasible for $1\le \ell\le K$ with high probability ($\ge 1-1/\mbox{poly}(nmLK/\epsilon)$). This implies that the {\sc Oracle} will not return infeasibility w.h.p.

To establish the correctness, it remains to bound the error in each constraint when the samples are drawn from $\DD$ and $\M(\vec{\alpha}[1,\ldots,K], \vec{\beta}[1,\ldots,K])$ is used determine allocations and payments. We focus on one constraint:
\begin{align*}
\frac{\sum_{\vec{t}\in S_{it_i\ell}} x_{iq}(\vec{t})}{|S_{it_i\ell}|} -X_{iq}(t_i) &\ge - \delta \qquad\forall i, t_i\in T_i
\end{align*}
Using Theorem~\ref{thm:ahk}, we get
\begin{align*}
\frac{1}{K}\times \sum_{1\le \ell\le K} \left(\frac{\sum_{\vec{t}\in S_{it_i\ell}} x^{\ell}_{iq}(\vec{t})}{|S_{it_i\ell}|} -X^{\ell}_{iq}(t_i) \right) &\ge -\delta-\delta
\end{align*}
where $x^{\ell}_{iq}(\vec{t})$ denotes whether in round $\ell$, buyer $i$ is given $q$ items when the type vector is $\vec{t}$. Applying Lemma~\ref{lem:whp} once again, we have:
\begin{align*}
\sum_{1\le \ell\le K}\frac{\sum_{\vec{t}\in \DD|t_i\in \vec{t}}\mu(\vec{t})x^{\ell}_{iq}(\vec{t})/f_i(t_i) -X^{\ell}_{iq}(t_i)}{K} &\ge -3\delta
\end{align*}
Rearranging, we obtain:
\begin{align*}
\sum_{\vec{t}\in \DD|t_i\in \vec{t}}\frac{\mu(\vec{t})}{f_i(t_i)} \times \sum_{1\le \ell\le K}\frac{x^{\ell}_{iq}(\vec{t})}{K}  - \sum_{1\le \ell\le K} \frac{X^{\ell}_{iq}(t_i)}{K} &\ge -3\delta
\end{align*}

This shows that if round $\ell \in \{1,2,\ldots,K\}$ is chosen uniformly at random and its allocation and payment according to mechanism $\WMP(\vec{\alpha}^{\ell}, \vec{\beta}^{\ell})$ are used, this yields a $\epsilon$-BIC mechanism with revenue $\OPT - \epsilon$.

\subsection{Handling Private Budgets}\label{sec:private-budget}
Next, we illustrate how to extend our techniques when buyers' budgets are private and specified through the reported types. Recall that $B_i(t_i)$ denotes buyer $i$'s budget when his reported type is $t_i$. Under the assumption that, the utility of paying more than the budget is $-\infty$ for buyer, we offer all items to an arbitrarily chosen buyer with a vanishingly small probability; otherwise we run the mechanism $\M\left(\vec{\alpha}[1\ldots T], \vec{\beta}[1\ldots T]\right)$. In such case, buyer $i$ never reports type $t'_i$ with $B_i(t'_i) > B_i(t_i)$ when his real type is $t_i$. Now the BIC constraint changes to
\[\begin{array}{rcl}
\sum_q v_i(q,t_i)X_{iq}(t_i) - P_i(t_i) &\ge& \sum_q v_i(q,t_i)X_{i}(t'_i) -P_i(t'_{i}) \\
& &\forall i, t_i,t'_i\in T_i {\mbox{ with }} B_i(t_i) \ge B_i(t'_i)
\end{array}\]
Thus, in-spite of having utility of $-\infty$ for paying more than the budget, the width of the LP does not increase and we can use techniques developed in the earlier part of this section to develop an $\epsilon$-BIC, ex-post IR mechanism.

\section{Overall Framework for Optimal Auction Design}\label{sec:framework}
We next summarize the generic framework for mechanism design using the multiplicative weight update method. There are two sets of variables:
\begin{itemize}
\item {\em Action Variables:} For each type vector $\vec{t}$, $x_i(\vec{t})$ and $p_i(\vec{t})$ indicates the mechanism's allocation and payment for buyer $i$. As discussed earlier, the optimal mechanism can take a randomized action for a type vector $\vec{t}$; it does not affect the analysis of the framework.
\item {\em Holistic Variables:} For buyer $i$, each $t_i,t_i'\in T_i$, $U_i(t_i, t'_i)$ indicates the average utility for reporting type $t_i$ when the real type is $t'_i$. Further, $\P_i(t_i, t'_i)$ is his average payment for reported type $t_i$ when the real type is $t'_i$.
\end{itemize}
There are three sets of constraints:
\begin{itemize}
\item $\F(\vec{t})$: the feasible set of actions for type vector $\vec{t}$. In each realization, $\left<\vec{x}(\vec{t}), \vec{p}(\vec{t})\right>\in \F(\vec{t})$.
\item $\expert$: the set of constraints that connect action variables to holistic variables. {\em This connection is problem specific}, we illustrate its generic structure.

For buyer $i$ and a reported type vector $\vec{t}$, where $t_i\in\vec{t}$ is his reported type,  let $\u_i(\vec{t}_{t_i\in \vec{t}}, t'_i, x_i(\vec{t}), p_i(\vec{t}))$  and $h_i(\vec{t}_{t_i\in \vec{t}}, t'_i, x_i(\vec{t}), p_i(\vec{t}))$ be his utility and payment when $t'_i$ is his real type  and $\left<x_i(\vec{t}), p_i(\vec{t})\right>$ are the mechanism's action variables.
 Then the constraints for buyer $i$ are
\begin{align*}
U_i(t_i, t'_i) &= \sum_{\vec{t} | t_i \in \vec{t}} \frac{\mu(\vec{t})}{f_i(t_i)} \u_i(\vec{t}_{t_i\in \vec{t}}, t'_i, x_i(\vec{t}), p_i(\vec{t}))\\
\P_i(t_i, t'_i) &= \sum_{\vec{t} | t_i \in \vec{t}} \frac{\mu(\vec{t})}{f_i(t_i)} h_i(\vec{t}_{t_i\in \vec{t}}, t'_i, x_i(\vec{t}), p_i(\vec{t}))
\end{align*}
\item Linear constraints on holistic variables such as BIC, revenue, and other ex-interim constraints.
\end{itemize}
Each equation in $\expert$ plays the role of an expert in the multiplicative weight update framework. The oracle problem {\sc Oracle}$(\vec{\alpha^{\ell}}, \vec{\beta^{\ell}})$ in round $\ell$ asks us to solve, for dual multipliers $\vec{\alpha^{\ell}}, \vec{\beta^{\ell}}$, the following:
\[\begin{array}{c}
{\mbox {Maximize}} \sum_{\vec{t}} \left( \mu(\vec{t})\times \sum_{i, t_i'\in T_i}\left(\alpha_{it_it'_i}^{\ell}\u_i(\vec{t}_{t_i\in \vec{t}}, t'_i, x_i(\vec{t}), p_i(\vec{t})) + \beta_{it_it'_i}^{\ell} h_i(\vec{t}_{t_i\in \vec{t}}, t'_i, x_i(\vec{t}), p_i(\vec{t})) \right)\right)
\end{array}\]
subject to constraint $\F(\vec{t})$ for each $\vec{t}$. The optimization problem $\WMP(\vec{\alpha^{\ell}}, \vec{\beta^{\ell}})$ for type vector $\vec{t}$ maximizes, subject to constraint $\F(\vec{t})$, the quantity
\[\begin{array}{c}
\sum_{i, t_i'\in T_i}\left(\alpha_{it_it'_i}^{\ell}\u_i(\vec{t}_{t_i\in \vec{t}}, t'_i, x_i(\vec{t}), p_i(\vec{t})) + \beta_{it_it'_i}^{\ell} h_i(\vec{t}_{t_i\in \vec{t}}, t'_i, x_i(\vec{t}), p_i(\vec{t})) \right)
\end{array}\]
Let $\tt ALG$ be the algorithm that can solve this problem optimally. Then, as established in Section~\ref{sec:multi-unit-budget}, it suffices to solve $\WMP(\alpha,\beta)$ for $C$ samples of type vectors drawn from the joint distribution; $C=poly(n, m, \sum_i |T_i|, 1/\epsilon, L)$. The dual multipliers $\left(\vec{\alpha^{\ell+1}}, \vec{\beta^{\ell+1}}\right)$ are computed based on the error in each constraint in round $\ell$. The process is repeated $K=poly(n, \sum_i |T_i|, 1/\epsilon, L)$ times; this generates $K$ sets of dual variables $\left((\alpha^1, \beta^1), (\alpha^2, \beta^2),\ldots,(\alpha^K, \beta^K)\right)$. The eventual mechanism is as follows: for reported type vector $\vec{t}$, pick a round at random, say $\ell$, and solve the optimization problem $\WMP(\vec{\alpha^{\ell}}, \vec{\beta^{\ell}})$ for type vector $\vec{t}$.

In general, a main hurdle to achieving tractability is that the Lagrange multipliers, $\vec{\alpha}, \vec{\beta}$ can be negative.
We note, in many cases the {\em full generality of the framework is not needed}: We can eliminate several variables, and the constraints $\expert$ as well as the holistic variables can take a simple form that is more tractable.

We present instantiations of this framework for multi-item auctions next. The key point is the following: {\em As long as the {\sc Oracle} subproblem constructed above is polynomial time solvable or admits to an approximation algorithm, the entire framework yields a polynomial time algorithm, or a similar approximation guarantee.}

\paragraph{\bf Multi-item Auctions and Tractability} It is quite straightforward to extend the above framework to multi-dimensional auctions. We discuss the case with budgets and other ex-post IR constraints such as envy-freeness in Appendix~\ref{sec:multi-item}, and we only summarize it below. The action variables are $x_{ij}(\vec{t})$, which denotes whether item $j$ allocated to buyer $i$ when the buyers' type vector is $\vec{t}$; and $p_{i}(\vec{t})$, which is the price paid by buyer $i$. For any $\vec{t}$, the allocation and payments must be feasible for a set of allocation and IR constraints, denoted by $\F(\vec{t})$. The holistic variables are $X_{ijt_i}$, denoting the expected allocation of item $j$ to buyer $i$ when his type is $t_i$, and $P_{it_i}$ denoting the expected price paid. When a buyer's welfare is linear in the valuations for different items, the BIC constraints are linear in the holistic variables. The overall problem reduces to checking feasibility of constraints
\[\begin{array}{rcll}
 f_i(t_i) X_{ijt_i} &=& \sum_{\vec{t} | t_i\in \vec{t}} \mu(\vec{t}) x_{ij}(\vec{t}) &  \forall i,j, t_i\in T_i \\
 f_i(t_i) P_{it_i} &=& \sum_{\vec{t} | t_i\in \vec{t}} \mu(\vec{t}) p_{i}(\vec{t}) & \forall i,t_i\in T_i \\
 \end{array} \]
In Appendix~\ref{sec:multi-item}, we show that when the items are divisible, the agent's utilities are linear, and the set of feasible allocations is convex, the {\sc Oracle} problem remains tractable even with private ex-post budgets and other IR constraints.

We note, if there are no ex-post constraints on payments (setting in~\cite{CDW_EC12,CDW_FOCS12,CDW_SODA13}), the only action variables capture allocations for each buyer given the joint type revealed, so that $\F(\vec{t})$ is simply the space of feasible allocations. The constraints $\expert$ relate these per-scenario allocation variables to the expected allocations for each buyer given his type.  Since there are no per-scenario payment constraints, the holistic payment variables are only present in the BIC and ex-interim IR constraints, and do not need corresponding action variables. Furthermore, when the allocation sets are downward closed, it is easy to show that the constraints $\expert$ can now be written as {\em inequality} constraints instead of equality as follows.
\[\begin{array}{rcll}
 f_i(t_i) X_{ijt_i} & \le & \sum_{\vec{t} | t_i\in \vec{t}} \mu(\vec{t}) x_{ij}(\vec{t}) &  \forall i,j, t_i\in T_i
  \end{array} \]
The advantage of doing this is that it makes the {\sc Oracle} subproblem have non-negative Lagrange multipliers, which makes the subproblem equivalent to welfare maximization over the feasible allocation space. We sketch in Appendix~\ref{sec:general_allocation_constraint} how this approach derives the results for approximately optimal auctions over downward closed allocation spaces presented in~\cite{CDW_SODA13}.

\section{Non-Linear Utility Functions}\label{sec:non-linear}
In this section, we consider multi-unit auctions (with budgets) when agents have non-linear utility functions. We first consider two well studied special cases of buyers with non-linear utility function over allocation and payment, namely (a) auction with quitting rights and (b) the soft budget constraint with cost to borrow. We also address a general setting in which the seller has a non-linear utility function over revenue and buyers have arbitrary utility functions over allocations and payments (which includes scenarios such as buyers' and/or seller's risk-aversion).

\subsection{Auctions with Quitting Rights}\label{sec:quitting_rights}
We first consider a model of auctions with {\em  quitting rights} ~\cite{cj}: The outcome of the mechanism is not binding on the buyer, and {\em in each realization} of the mechanism, the buyer has the option to quit if his payment is more than his valuation of the allocated set. In other words, in each outcome, his utility is the maximum of zero and the difference in his welfare and payment. We note, even though the ex-post IR constraint makes a buyer's utility non-negative for the correctly reported type, the non-binding nature of contract makes utility non-negative even for mis-reporting the type. This makes the BIC constraint non-linear in the expected allocations and payments, and a buyer's expected utility for each type needs to be computed explicitly.

As before, consider a realization of the mechanism, let $\vec{t}$ be the reported type vector of all buyers. In the realization, let $x_{iq}(\vec{t})\in\{0, 1\}$ be an indicator variable indicating whether buyer $i$ is allocated $q\in \{0,1,\ldots,m\}$ items when the reported type vector is $\vec{t}$. Let $p_{i}(\vec{t})$ be the payment made by buyer $i$.
The variables should satisfy the constraints $\F(\vec{t})$ below.
\[ \begin{array}{rcll}
\sum_{i} qx_{i,q}(\vec{t}) & \le & m & \\
\sum_{q} v_i(q,t_i)x_{iq}(\vec{t}) - p_{i}(\vec{t})   &\ge&  0 & \forall i\\
x_{iq}(\vec{t}) &\in& \{0,1\} &\forall i,q\\
\sum_q x_{iq} &=& 1 & \forall i \\
p_{i}(\vec{t}) &\in& [0, B_i] &\forall i
\end{array}\]
Let $U_{i}(t_i, t'_i)$ be buyer $i$'s expected utility when his reported type is $t_i$ and the real type is $t'_i$.
We get, $\forall i, t_i, t'_i\in T_i$
\begin{align*}
U_i(t_i, t'_i) = \sum_{\vec{t} | t_i \in \vec{t}} \frac{\mu(\vec{t})}{f_i(t_i)} \left(\max\left\{\sum_q v_{iq}(t'_i)x_{iq}(\vec{t}) -p_{i}(\vec{t}), 0\right\}\right)
\end{align*}
(Recall, we denote buyer $i$'s reported type in $\vec{t}$ by $t_i$.) The BIC constraint is the following:
\begin{equation*}
U_{i}(t_i, t_i) \ge U_{i}(t'_i, t_i) \qquad \forall i, t_i, t'_i\in T_i
\end{equation*}
We add the revenue constraints:
\[\begin{array}{rcll}
\sum_{i,t_i} f_i(t_i)  P_i(t_i) &\ge&  \OPT\\
f_i(t_i) P_{i}(t_i) &=& \sum_{\vec{t} | t_i \in \vec{t}} \mu(\vec{t}) p_{i}(\vec{t}) &\forall i, t_i\in T_i
\end{array}\]
We discuss the correctness of the revenue constraints: as the mechanism is always individual rational for the reported type, the revenue from buyer $i$ for type vector $\vec{t}$ is $p_i(\vec{t})$, when he reports his type truthfully.
The {\sc Oracle} problem $\WMP(\vec{\alpha}, \vec{\beta})$, which for every type vector $\vec{t}$, solves the following subject to constraints $\F(\vec{t})$:
\[\begin{array}{rl}
\mbox{Maximize } &\sum_{it'_i} \alpha_{it_it'_i}\left( \max\left\{\sum_{q} v_i(q,t'_i)x_{iq}(\vec{t}) - p_{i}(\vec{t}) , 0\right\}\right) +\sum_i \beta_{it_i}p_{i}(\vec{t})
\end{array}\]
For any type vector $\vec{t}$, this can be solved using dynamic programming: if $q$ items are allocated to buyer $i$ then the payment is uniquely determined as follows:

\smallskip
\noindent
(a) If $\beta_{it_i}\le 0$, then $p_i(\vec{t})=0$\\
(b) If $\beta_{it_i} > 0$, then $p_i(\vec{t})$ is the price that maximizes \\
$\sum_{t'_i} \alpha_{it_it'_i}\left( \max\{\sum_{q} v_i(q,t'_i)x_{iq}(\vec{t}) - p_{i}(\vec{t}) , 0\}\right)$. Its value lies in the set $\{0, B_i\}\cup\{v_i(q,t'_i)|t'_i\in T_i\}$ and can be computed in polynomial time. This implies a poly-time dynamic program analogous to that in Section~\ref{sec:multi-unit-budget}. We get the following theorem using multiplicative weight update method.
\begin{theorem}
For multi-unit auctions with quitting rights and budget constraints, there exists a polynomial time algorithm to compute an $\epsilon$-BIC mechanism with revenue $\OPT-\epsilon$.
\end{theorem}

\subsection{Soft Budget Constraints}
Until now, we have assumed that buyers have fixed budgets and cannot pay more than their budgets. Recently, models with {\em soft} budgets have been considered assuming a buyer's ability to borrow at an additional cost~\cite{DHW11}. The model is formally defined as follows: buyer $i$ is associated with a piecewise linear function $c_i:\R\rightarrow \R$; $c_i(p)$ denotes the total cost for arranging payment $p$ and we assume that $c_i(p)-p$ is monotone. For instance, if buyer $i$ can arrange for interest-free $B_i$ dollars, and has to pay an interest rate of $r$ beyond that; then $c_i(p) = p +\max\{(1+r)\times(p-B_i), 0\}$. The function $c_i$ easily accounts for situations where (a) there is an additional fixed cost to setup the loan, and (b) the total interest amount increases with the amount of the loan.

The buyer's utility in a realization is his valuation of the allocated items minus the cost for his payment. The details are similar to the previous setting, we note the differences. The function $U(t_i,t'_i)$ changes to
\[\begin{array}{c}
f_i(t_i)U_i(t_i,t'_i)=\sum_{\vec{t}|t_i\in \vec{t}}\mu(t)\left(\sum_qv_{i}(q, t'_i)x_{iq}(\vec{t})-c_i(p_i(\vec{t}))\right)
\end{array}\]
The additional constraint in $\F(\vec{t})$ is the IR constraint that the mechanism is rational for the reported type:
\[\begin{array}{c}
\sum_qv_{i}(q, t_i)x_{iq}(\vec{t})-c_i(p(\vec{t}))\ge 0
\end{array}\]
As before, for each type vector $\vec{t}$, the {\sc Oracle} solves the following subject to constraints $\F(\vec{t})$:
\[\begin{array}{c}
\mbox{Maximize } \sum_{it'_i} \alpha_{it_it'_i}\left( \sum_{q} v_i(q,t'_i)x_{iq}(\vec{t}) - c_i(p_{i}(\vec{t}))\right)
\end{array}\]
The optimal solution has the following property: if $q$ items are allocated to buyer $i$, then his payment within the set $\{0, v_i(q,t_i)\}$ and among the end points of the linear pieces of $c_i(p)$. Thus the optimal assignment can be found using dynamic programming and an $\epsilon$-BIC mechanism can be designed using the multiplicative update framework.

\subsection{Non Linear Utility for Seller}
The settings considered until now assume that the seller's utility is linear in revenue and the objective is to maximize his expected revenue. In this section, we address the setting in which the seller is associated with an arbitrary monotone utility function $\U:\R\rightarrow\R$; {\em i.e.} the seller's utility for revenue $Z$ is $\U(Z)$. The objective of the mechanism is to maximize the seller's expected utility $\E_{Z}\left[\U(Z)\right]$. Further, each buyer $i$ is associated with an arbitrary utility function $\u_i:\Z\times\N\times T_i\rightarrow \R^+$; buyer $i$'s utility for receiving $q$ items for a payment of $p$ is $\u_i(p,q,t_i)$ where $t_i$ is his type. This includes scenarios where the seller and/or the buyers are risk averse.

We make following assumptions on the nature of utility functions, as well as on the optimal mechanism: (a) a buyer's payment is allowed to be negative (positive transfer is allowed) however the set of possible payments from a buyer is restricted to integers between $-L$ and $L$, (b) for each $i, t_i\in T_i, p,q$, $|\u_i(p,q,t_i)|\le L$, and (c) $\U(0)=0$.

The assumption about integral payments is for simplicity in exposition; if all buyers' and seller's utility functions are monotone and have bounded derivative, then any $BIC$ mechanism can be converted into an $\epsilon$-BIC mechanism whose payments are integral - we simply round down the payment to a suitably discretized value so that the utility does not change by more than $\epsilon$.  The reason we allow positive transfers in the above framework is that it can raise much more utility for the seller (see~\cite{BCK12} for details).

\smallskip
In the optimal mechanism, let $x_{ipq}(\vec{t})$ be an indicator variable that denote whether buyer $i$ is given $q$ items for a payment of $p$ when buyers' type vector is $\vec{t}$. For any $\vec{t}$, allocation and payments must be feasible for the following set of constraints $\F(\vec{t})$, where the final constraint imposes ex-post IR:
\[ \begin{array}{rcll}
\sum_{i,p,q} q x_{ipq}(\vec{t}) & \le & m & \\
\sum_{p,q} x_{ipq}(\vec{t}) & = & 1 & \forall i \\
x_{ipq}(\vec{t}) & \in & \{0,1\} & \forall i,p,q \\
x_{ipq}(\vec{t}) & = & 0 & \forall i, p,q \mbox{ s.t. } \u_i(p,q,t_i) < 0
\end{array}\]
Let $X_{ipq}(t_i)$ denote the probability that buyer $i$ is allocated $q$ items at price $p$ when his revealed type is $t_i$.  The problem reduces to checking the feasibly of the following set of constraints
\[ \begin{array}{rcll}
f_i(t_i) X_{ipq}(t_i) - \sum_{\vec{t} | t_i \in \vec{t}} \mu(\vec{t}) x_{ipq}(\vec{t}) &=&0 &\forall i, t_i\in T_i, p, q\\
\sum_{\vec{t}}\mu(\vec{t})\U\left(\sum_{i,p,q}p\times x_{ipq}(\vec{t})\right) &\ge& R
\end{array}\]
subject to constraints given by $\F(\vec{t})$ (for all type vectors $\vec{t}$) and the BIC constraint given below:
\[ \begin{array}{rcll}
\sum_{p,q} \left(X_{ipq}(t_i)-X_{ipq}(t'_i)\right)\times \u_i(p, q, t_i) &\ge& 0 & \forall i, t_i,t'_i\in T_i\\
\end{array}\]
As before, the {\sc Oracle} problem $\WMP(\vec{\alpha}, \vec{\beta})$, for each type vector $\vec{t}$, solves the following
\[ \begin{array}{c}
\mbox{Max} \left(\sum_{i,p,q,t_i\in \vec{t}} \alpha_{ipqt_i} x_{ipq}(\vec{t}) + \beta \times \U\left(\sum_{i,p,q}p\times x_{ipq}(\vec{t})\right)\right)
\end{array}\]
subject to constraints $\F(\vec{t})$. We note that $\alpha_{ipqt}$ can be both positive or negative, however $\beta$ is always positive.

\smallskip
\noindent
{\em Solving $\WMP(\vec{\alpha}, \vec{\beta})$ for $\vec{t}$:} The problem can be solved optimally using a dynamic program as follows: let $A(i, k, z)$ be the maximum value of the objective function when $k$ items are allocated to first $i$ buyers generating a total revenue of $z$ from them. The recurrence relation for the dynamic program is
\[\begin{array}{rl}
A(i,k,z)=&\min_{q \le m-k}\max_p\{\beta\times\left(U(z+p)-U(z)\right)+A(i+1, k+q, z+p) + \alpha_{ipqt_i}  \}
\end{array}\]
Note, our assumption that payments to buyers are always integral and bounded between $-L$ and $+L$. Thus, $A(*,*,z)$ needs to defined only for $2nL$ different values of revenue, and the dynamic program can be solved in time polynomial in $L$.
Using the multiplicative weight update framework, we get following theorem:
\begin{theorem}\label{thm:risk_aversion}
Given a seller with arbitrary utility function over total revenue and buyers with arbitrary utility functions over payments and allocations,
there exists a polynomial time algorithm to compute an $\epsilon$-BIC mechanism with expected utility $\OPT-\epsilon$ for the seller.
\end{theorem}
We note that if we only wish to impose ex-interim IR, then we remove the final constraint in $\F(\vec{t})$, and instead add:
\[ \begin{array}{rcll}
\sum_{p,q} X_{ipq}(t_i) \times \u_i(p, q, t_i) & \ge & 0 &\forall i, t_i\in T_i\\
\end{array}\]

\newpage
\appendix
\section{Multi-item Auctions}
\label{sec:multi-item}
In this section, we illustrate how to use our framework for multi-item auctions, when buyers have additive valuations and items are divisible. The setting is described in detail in Section~\ref{sec:prelim}. As earlier, the objective is to design an $\epsilon$-BIC mechanism that is feasible, ex-post IR, satisfies buyers' budget constraints, and maximizes seller's revenue.
For simplicity in exposition, we assume that buyers' budgets are public; our techniques can be easily extended to handle private budgets. The technically interesting aspect is that we can write constraints coupling allocations and payments across buyers, for instance {\em envy-freeness} constraints. At the end, in Section~\ref{sec:general_allocation_constraint}, we show that our framework provides a simple recipe for deriving recent results in \cite{CDW_EC12,CDW_FOCS12,CDW_SODA13} for general allocation constraints.

In the mechanism, let $x_{ij}(\vec{t})$ denote the fraction of item $j$ allocated to buyer $i$ when the buyers' type vector is $\vec{t}$; and let $p_{i}(\vec{t})$ be the price paid by buyer $i$. For any $\vec{t}$, the allocation and payments must be feasible for the following convex set of constraints, denoted by $\F(\vec{t})$:
\begin{align*}
&\vec{x}(\vec{t}), \vec{p}(\vec{t}) \in \mathcal{P}(\vec{t}) \qquad \mbox{and} \qquad p_{i}(\vec{t}) \le B_i   \ \ \forall i \qquad \mbox{and}\\
&\sum_j v_{ij}(t_i) x_{ij}(\vec{t}) \ge p_i(\vec{t}) \qquad \forall i, t_i\in T_i
\end{align*}
Let $X_{ijt_i}$ denote the expected allocation of item $j$ to buyer $i$ when his type is $t_i$ and let $P_{it_i}$ denote the expected price paid. For target revenue $\OPT$, the problem reduces to checking feasibility of constraints:
\[\begin{array}{rcll}
 f_i(t_i) X_{ijt_i} &=& \sum_{\vec{t} | t_i\in \vec{t}} \mu(\vec{t}) x_{ij}(\vec{t}) &  \forall i,j, t_i\in T_i \\
 f_i(t_i) P_{it_i} &=& \sum_{\vec{t} | t_i\in \vec{t}} \mu(\vec{t}) p_{i}(\vec{t}) & \forall i,t_i\in T_i \\
 \end{array} \]
subject to the constraints given by $\F(\vec{t})$ for all $\vec{t}$ and:
\[ \begin{array}{rl}
\sum_{j} v_{ij}(t_i) X_{ijt_i} - P_{it_i} &\ge  \sum_{j} v_{ij}(t_i) X_{ijt'_i} -  P_{it'_i}  \ \  \forall i, t_i,t'_i\in T_i \\
\sum_{i,t} f_i(t_i) P_{it_i} & \ge  \OPT \\
\end{array}\]
We use techniques developed in Section~\ref{sec:multi-unit-budget} to solve the feasibility problem. As before, it would require implementation of {\sc Oracle}$(\vec{\alpha}, \vec{\beta})$ that reduces to solving $\WMP(\alpha, \beta)$, which for each type vector $\vec{t}$, computes:
$$  \mbox{Maximize} \sum_{i,j} \alpha_{ijt_i} x_{ij}(\vec{t}) + \sum_{i} \beta_{it_i} p_{i}(\vec{t})$$
subject to the constraints $\F(\vec{t})$. As $\F(\vec{t})$ is convex, $\WMP(\vec{\alpha}, \vec{\beta})$ can be solved optimally using an LP for any $\vec{t}$. Thus we get the following theorem:
\begin{theorem}\label{thm:multi-item}
For multi-item auctions with additive valuations and arbitrary convex constraints on the set of feasible allocations an payment in any realization, there exists a polynomial time algorithm to compute a  feasible, ex-post rational, $\epsilon$-BIC mechanism with expected revenue $\OPT-\epsilon$.
\end{theorem}

\paragraph{\bf Envy-Freeness} The techniques in this section can be extended to support envy-freeness in each realization of the mechanism. Consider the setting where $\F(\vec{t})$ only has the ex-post IR constraints
\begin{align*}
\vec{x}(\vec{t}), \vec{p}(\vec{t}) \ge \vec{0} \qquad \mbox{and} \qquad \sum_j v_{ij}(t_i) x_{ij}(\vec{t}) \ge p_i(\vec{t}) \qquad \forall i
\end{align*}
{\em Envy-freeness} means that no buyer receives more utility from the allocation and payment made to another buyer. We can enforce this by simply adding
\begin{equation*}
\sum_{j}x_{ij}(\vec{t})v_{ij}(t_i)-p_i(\vec{t}) \ge \sum_{j}x_{i'j}(\vec{t})v_{ij}(t_i)-p_{i'}(\vec{t}) \qquad  \forall i,i'
\end{equation*}
to $\F(\vec{t})$. Note, $\F(\vec{t})$ still remains convex and $\WMP(\vec{\alpha}, \vec{\beta})$ can be solved in polynomial time.

Our auctions extend to the case where there are $m = O(1)$ indivisible items. The valuation to buyer $i$ for obtaining subset $S$ of items given type $t$ is $v_i(S,t)$.

\subsection{General Allocation Constraints}\label{sec:general_allocation_constraint}
For multi-dimensional auctions without ex-post IR or payment constraints, our approach can handle general allocation constraints (over indivisible items) as considered in \cite{CDW_EC12,CDW_FOCS12,CDW_SODA13}. We illustrate this for the multi-item auction problem with additive valuations and indivisible items, where the allocation set $\F(\vec{t}) =\{ \vec{x}(\vec{t}) \in \mathcal{P}\}$ is downward closed and only admits to $c$-approximately optimizing welfare~\cite{CDW_SODA13}. The key hurdle with applying the above methodology as is, is that the dual multipliers can be negative. But this is easy to fix: Since we do not have payment constraints, we simply need to check the feasibility of the following {\em inequality} constraint:
\[\begin{array}{rcll}
 f_i(t_i) X_{ijt_i} & \le & \sum_{\vec{t} | t_i \in \vec{t}} \mu(\vec{t}) x_{ij}(\vec{t}) &  \forall i,j, t_i \\
 \end{array} \]
 subject to the constraints given by $\F(\vec{t})$ for all $\vec{t}$ and the BIC, ex-interim IR, and revenue constraints on $\{X_{ijt}, P_{it}\}$.

Consider a relaxed feasibility problem for $c \ge 1$:
\[\begin{array}{rcll}
 f_i(t_i) X_{ijt_i} & \le & c  \times \sum_{\vec{t} | t_i \in \vec{t}} \mu(\vec{t}) x_{ij}(\vec{t}) &  \forall i,j, t_i \\
 \end{array} \]
The {\sc Oracle}$(\alpha)$ problem maximizes welfare, $\sum_{i,j} \alpha_{ijt_i} x_{ij}(\vec{t})$ subject to $\F(\vec{t})$. Since the constraint is inequality, we have $\alpha \ge 0$ so that the {\sc Oracle} sub-problem is precisely welfare maximization. Suppose this has a $c$ approximation. Suppose there is a  feasible solution to the original program of the following form
\[\begin{array}{rcll}
 f_i(t_i) X^*_{ijt_i} & \le & \sum_{\vec{t} | t_i \in \vec{t}} \mu(\vec{t}) x^*_{ij}(\vec{t}) &  \forall i,j, t_i \\
 \end{array} \]
Then, we must have:
\[\begin{array}{rcl}
\sum_{i,j,t_i} \alpha_{ijt_i} f_i(t_i) X^*_{ijt_i} & \le & \sum_{i,j,t_i} \alpha_{ijt_i} \sum_{\vec{t} | t_i \in \vec{t}} \mu(\vec{t}) x^*_{ij}(\vec{t})
 \end{array} \]
Since {\sc Oracle}$(\alpha)$ has a $c$-approximation, there exists a solution $\tilde{x}$ that can be computed in polynomial time such that:
\[\begin{array}{rcl}
\sum_{i,j,t_i} \alpha_{ijt_i} f_i(t_i) X^*_{ijt_i} & \le & c \times \sum_{i,j,t_i} \alpha_{ijt_i} \sum_{\vec{t} | t_i \in \vec{t}} \mu(\vec{t}) \tilde{x}_{ij}(\vec{t})
 \end{array} \]
This implies that the multiplicative weight update method using the $c$-approximate oracle will never declare infeasibility of the relaxed program. The rest of the argument is the same as before, and yields a randomized mechanism using the multiplicative weight method.

The only issue is in making the final mechanism obey the BIC constraints since the inequality constraint implies the probability of allocation need not be exactly $X_{ijt_i}$ in expectation, and could be larger. The fix is however a standard trick that works for downward closed environments. Suppose the final slack in the constraint for $(i,j,t_i)$, that is, the ratio of LHS to RHS, is $\beta_{ijt_i} \le 1$, then whenever the mechanism decides to allocate item $j$ to buyer $i$ for type $t_i$, we allocate with probability $\beta_{ijt_i}$. This ensures the expected allocation is {\em exactly} $ X_{ijt_i}/c$. If the charge exactly $P_{it}/c$ via an all-pay scheme, then this ensures that the mechanism is BIC with revenue $\OPT/c$.

\section{Correlated Distributions}\label{sec:correlated}
In this section, we illustrate how to extend our techniques to handle correlated distributions for multi-unit auctions. Our techniques are generic and easily extend to other settings considered in the paper. Note, in Section~\ref{sec:multi-unit-budget}, the BIC constraint of multi-unit auctions is stated as
\[\begin{array}{rcl}
\sum_{q}v_i(q,t_i)X_{iq}(t_i) - P_{i}(t_i) &\ge&  \sum_q v_i(q,t_i)X_{iq}(t'_i) - P_{i}(t'_i) \\
& &\forall i, t_i, t'_i\in T_i
\end{array}\]
where $X_{iqt_i}$ and $P_{it_i}$ are defined as
\[\begin{array}{rcll}
f_i(t_i) X_{iqt_i} &=& \sum_{\vec{t} | t_i\in \vec{t}} \mu(\vec{t}) x_{iq}(\vec{t}) &\forall i, q,t_i\in T_i \\
f_i(t_i) P_{it_i} &=& \sum_{\vec{t} | t_i\in \vec{t}} \mu(\vec{t}) p_{i}(\vec{t}) &\forall i, t_i\in T_i \\
\end{array}\]
The BIC constraint stated above requires independence across buyers. With correlated distributions, if buyer $i$'s real type is $t_i$ and his reported type is $t'_i$, then the probability that he is allocated $q$ items is not necessarily $X_{iq}(t'_i)$; same holds for expected payments. With correlated distributions, these values need to computed explicitly. Fix buyer $i$; for each $t_i,t'_i\in T_i$, define $X_{iqt_it'_i}$ as the probability that buyer $i$ is allocated $q$ items when his real type is $t_i$ and the reported type is $t'_i$; $P_{it_it'_i}$ is defined as buyer $i$'s expected payment when his real type is $t_i$ and reported type is $t'_i$. For each type vector $\vec{t_{-i}}$ of other buyers, let $z_{it_it'_i}(\vec{t_{-i}})=\frac{\mu(t_i|\vec{t}_{-i})}{\mu(t'_i|\vec{t}_{-i})}$. We now simplify the value of $X_{iqt_it'_i}$.
\[\begin{array}{rcl}
X_{iqt_it'_i} &=& \sum_{\vec{t}_{-i}}\mu(\vec{t}_{-i}|t_i) \times x_{iq}(t'_i,\vec{t}_{-i})\\
&=& \sum_{\vec{t}|t'_i \in \vec{t}}\frac{\mu(\vec{t}_{-i}|t_i)\times \mu(\vec{t})}{\mu(\vec{t}_{-i}|t'_i)\times f_{i}(t'_i)}\times x_{iq}(\vec{t}) \\
&=& \sum_{\vec{t}|t'_i \in \vec{t}} \frac{\mu(t|\vec{t}_{-i})}{\mu(t'|\vec{t}_{-i})}\times \mu(\vec{t})\times x_{iq}(\vec{t}) \\
&=& \frac{1}{f_i(t_i)}\times \sum_{\vec{t}|t'_i \in \vec{t}}z_{it_it'_i}(\vec{t}_{-i})\mu(\vec{t})x_{iq}(\vec{t})
\end{array}\]
Note, buyer $i$'s type in $\vec{t}$ is $t'_i$. Thus the problem reduces to checking the feasibility of following constraints, denoted by $\expert$
\[\begin{array}{c}
f_i(t_i)X_{iqt_it'_i} = \sum_{\vec{t}|t'_i \in \vec{t}}z_{it_it'_i}(\vec{t}_{-i})\mu(\vec{t})x_{iq}(\vec{t}) \ \ \forall i, q, t_i,t'_i \in T_i\\
f_i(t_i)P_{it_it'_i} = \sum_{\vec{t}|t'_i \in \vec{t}}z_{it_it'_i}(\vec{t}_{-i})\mu(\vec{t})p_i(\vec{t}) \ \ \ \ \forall i, t_i,t'_i \in T_i
\end{array}\]
subject to constraints $\F(\vec{t})$ for all $\vec{t}$, and:
\[\begin{array}{rcll}
\sum_q X_{iqt_it_i}v_{i}(q,t_i) - P_{it_it_i} &\ge& \sum_q X_{iqt_it'_i}v_{i,q}(t_i) - P_{it_it'_i} \\
 & &\forall i, t_i,t'_i \in T_i
\end{array}\]
\[\begin{array}{rcll}
\sum_{i,t_i\in T_i} f_{i}(t_i)P_{it_it'_i} &\ge& R &\forall i, t_i\in T_i \\
\end{array}\]
As before, the implementation of {\sc Oracle}$(\vec{\alpha},\vec{\beta})$ requires solving $\WMP(\vec{\alpha}, \vec{\beta})$, which for each vector $\vec{t} = (t_i, \vec{t}_{-i})$ of types, maximizes:
$$\sum_{i,q,t'_i} \alpha_{iqt'_it_i} z_{it'_it_i}(\vec{t}_{-i})x_{iq}(\vec{t}) +  \sum_{i,t'_i} \beta_{it'_it_i} z_{it'_it_i}(\vec{t}_{-i})p_i(\vec{t})$$
subject to constraints $\F(\vec{t})$. This problem can be solved optimally using a dynamic program just as before.

\medskip
\noindent
{\em Running Time:} Note that the width $\rho$ of the constraint $\expert$ is linearly dependent upon $\max_{i,t_i,t'_i \in T_i}z_{it_it'_i}(\vec{t}_{-i})$, denoted by $z_{\max}$. The term $z_{max}$ can be seen as the maximum ratio of conditional probabilities of two different types for a buyer, when the type vector of other buyers is fixed. 
\begin{theorem}
For multi-unit auctions with budgeted buyers and correlated distributions, if for any buyer $i$, conditioned of each type vector of other buyers $\vec{t}_{-i}$, if the ratio of maximum and minimum probability of his two types is bounded by $z_{max}$, then there exists a polynomial (in $z_{max}$) time algorithm to compute a feasible, ex-post rational, $\epsilon$-BIC mechanism with expected revenue $\OPT-\epsilon$.
\end{theorem}

\section{Budget Feasible Mechanisms}\label{sec:budget-feasible}
In this section, we illustrate our results for designing budget feasible mechanisms. We first describe the setting: the auctioneer has a (hard) budget of $B$; we assume that $B$ is an integer bounded by $L$. There are $n$ agents, and each agent produces a different type of item. Each agent is associated with a private cost to produce the item; specifically, agent $i$'s cost of producing his item is $c_i(t)$ where $t$ is his type. The auctioneer's valuation of set of items $S$ is given by a $\G(S) =\sum_{j\in S}v_j$, {\em i.e.}, linear in set $S$. The objective is to design a BIC mechanism that maximizes auctioneer's expected valuation of procured items such that (a) the auctioneer never pays more than $B$ in any realization, and (b) the mechanism is ex-post IR for each agent, i.e. if the item from agent $i$ of type $t$ is procured, then he is paid at least $c_i(t)$.
For simplicity in exposition, we assume $\{c_i(t_i)\}$ are integers in $\{0,1,2,\ldots, L\}$.


Let $\vec{t}$ be the reported valuation vector. In the realization of the mechanism, let $x_{i}(\vec{t})$ denote whether the item is procured from agent $i$, and $p_i(\vec{t})$ be the payment made to agent $i$. The allocation and payments must be feasible for the following set of constraints, denoted by $\F(\vec{t})$:
\[\begin{array}{rcll}
\sum_i p_i(\vec{t}) &\le& B \\
x_{i}(\vec{t})c_{i}(t_i) &\le& p_i(\vec{t}) &\forall i\\
x_{i}(\vec{t}) &\in& \{0,1\} &\forall i\\
p_{i}(\vec{t}) &\in& [0, B] &\forall i
\end{array}\]
Let $X_{it_i}$ be the probability that agent $i$'s item is procured when his type is $t_i$; and $P_{it_i}$ be the expected payment made to him in this case.  Let $S(\vec{t})$ denote the set of agents whose items are procured when the type vector is $\vec{t}$. Let $\OPT$ be the auctioneer's expected valuation of the procured items. Then the problem reduces to checking feasibility of constraints:
\[\begin{array}{rcll}
 f_i(t_i) X_{it_i} &=& \sum_{\vec{t} | t_i\in \vec{t}} \mu(\vec{t}) x_{i}(\vec{t}) &  \forall i,j, t_i\in T_i \\
 f_i(t_i) P_{it_i} &=& \sum_{\vec{t} | t_i \in \vec{t}} \mu(\vec{t}) p_{i}(\vec{t}) & \forall i,t_i \in T_i \\
 \sum_{\vec{t}}\mu(\vec{t})\sum_i v_ix_{i}(\vec{t}) &\ge& \OPT
 \end{array} \]
subject to constraints $\F(\vec{t})$ for all $\vec{t}$ and:
\[ \begin{array}{rcll}
P_{it_i} - c_{i}(t_i) X_{it_i} &\ge&  P_{it'_i}-c_{i}(t_i) X_{it'_i}  & \forall i, t_i, t'_i\in T_i \\
\end{array}\]

As in Section~\ref{sec:multi-unit-budget}, we require implementation of {\sc Oracle}$(\vec{\alpha},\vec{\beta},\gamma)$ which involves solving $\WMP(\vec{\alpha}, \vec{\beta}, \gamma)$, where for each vector $\vec{t}$ of types, it computes:
\begin{equation}
\mbox{Max} \left(\sum_{i} \alpha_{it_i} x_{i}(\vec{t}) + \sum_{i} \beta_{it_i} p_{i}(\vec{t}) + \gamma\sum_i v_ix_{i}(\vec{t})\right)
\end{equation}
subject to the constraints $\F(\vec{t})$. As earlier, $\alpha_{it_i}$ and $\beta_{it_i}$ can be positive or negative, however $\gamma$ is always positive. We now illustrate a polynomial time algorithm to solve this problem. Recall, $S(\vec{t})$ is the set of agents whose items are procured when the bid vector is $\vec{t}$. Then for each $i\in S(\vec{t})$, we have

\smallskip
\noindent
(a) $\beta_{it_i}\le 0\Rightarrow p_{i}(\vec{t})=c_{i}(t_i)$, and

\smallskip
\noindent
(b) If there is at least one agent $i\in S(\vec{t})$ such that $\beta_{it_i}\ge 0$, then there exists at most one agent $j={\mbox{argmax}}_{i\in S(\vec{t})}\beta_{it_i}$ such that $p_{j}(\vec{t})>c_{j}(t_j)$, and then $p_{j}(\vec{t}) = B-\sum_{i\in S(\vec{t}), i\neq j}c_{i}(t_i)$; such an agent can be seen as the {\em winner} in the allocation.

\medskip
Clearly, if agent $i$ is the winner of an allocation, then all agents $j\in S(\vec{t})$ must have $\beta_{jt_j}\le \beta_{it_i}$. The algorithm is as follows: We guess one agent as the winner; let $i$ be that agent. Let $S_i = \{j | \beta_{jt_j}\le \beta_{it_i}\}$. The goal is to find a subset $Q \subseteq S_i$ with the following properties:
\begin{itemize}
\item $i \in Q$ since we guess $i$ as belonging to $Q$.
\item $\sum_{j \in Q} c_j(t_j) \le B$, since $Q$ must be budget feasible.
\item {\sc Oracle}  $=\sum_{j \in Q, j \neq i} \left(\alpha_{jt_j} +  (\beta_{jt_j} - \beta_{it_i}) c_{j}(t_j) + \gamma v_j\right) + \alpha_{it_i} + \gamma v_i + B \beta_{it_i}$ is maximized.
\end{itemize}

This is now just a knapsack problem where the $c_j(t_j)$ are item sizes; $B$ is the knapsack capacity; and the profit of item $j$ is $\alpha_{jt_j} +  (\beta_{jt_j} - \beta_{it_i}) c_{j}(t_j) + \gamma v_j$. Since we assumed $c_j(t_j)$ are integers in $\{0,1,2,\ldots, L\}$, this can be solved in time polynomial in $n,L$. Thus we get the following theorem.
\begin{theorem}
When the auctioneer has additive valuation for the set of procured items, then there exists a poly-time algorithm to compute an ex-post IR, $\epsilon$-BIC mechanism with expected revenue $OPT-\epsilon$.
\end{theorem}
When the costs $c_i(\vec{t})$ are not integers, we can discretize them for the purpose of solving the knapsack instance. This leads to a violation of the budget $B$ by an additive amount $\epsilon$, but otherwise our result remains unchanged.

Our results also extends to the setting where the auctioneer is interested in procuring $m = O(1)$ items, and has arbitrary valuations $\G(S)$ on subsets $S$ of items of constant size.

\end{document}